\documentclass[a4paper]{llncs}

\usepackage{latexsym}
\usepackage{amssymb}
\usepackage{amsmath}
\usepackage{graphicx,subfigure}
\usepackage{pstricks,pst-node,pst-tree}
\usepackage{haskell}
\usepackage{relsize}
\usepackage{cite}

\input xy

\xyoption{all} 

\usepackage{listings}
\lstnewenvironment{code}
    {\lstset{}%
      \csname lst@SetFirstLabel\endcsname}
    {\csname lst@SaveFirstLabel\endcsname}
\newtheorem{thm}{Theorem}
\newtheorem{lem}[thm]{Lemma}
\newtheorem{prop}[thm]{Proposition}
\newtheorem{defn}[thm]{Definition}

\newcommand{\bi}{\begin{array}[t]{@{}l@{}}}
\newcommand{\ei}{\end{array}}
\newcommand{\ba}{\begin{array}}
\newcommand{\ea}{\end{array}}
\newcommand{\bda}{\[\ba}
\newcommand{\eda}{\ea\]}
\newcommand{\bp}{\begin{quote}\tt\begin{tabbing}}
\newcommand{\ep}{\end{tabbing}\end{quote}}

\newcommand{\EE}{{\cal E}}

\newcommand{\tlabel}[1]{\mbox{(#1)}}
\newcommand{\fig}[3]
        {\begin{figure*}[t]#3\
        \caption{\label{#1}#2}\  \end{figure*}}

\newcommand{\myirule}[2]{{\renewcommand{\arraystretch}{1.2}\ba{c} #1
                      \\ \hline #2 \ea}}

\newcommand{\multiset}[1]{\{\!\{#1\}\!\}}

\newcommand{\varsEE}[1]{\mathit{vars}(#1)}

\newcommand\Angle[1]{\langle#1\rangle}
\newcommand\Config[2]{\Angle{#1, #2}}
\newcommand{\solveRel}[2]{#1 \Rightarrow #2}
\newcommand{\solveRels}[2]{#1 \Rightarrow^* #2}

\newcommand{\solve}[1]{\mbox{\it solve}(#1)}

\newcommand{\regEq}{\approx}

\newcommand{\turns}{\, \vdash \,}

\newcommand{\sgap}{\quad}

\newcommand{\mathem}{\sf}

\newcommand{\REC}{\mbox{\mathem rec}}

\newcommand{\CASE}{\mbox{\mathem case}}

\newcommand{\ELSE}{\mbox{\mathem else}}
\newcommand{\IF}{\mbox{\mathem if}}
\newcommand{\OF}{\mbox{\mathem of}}
\newcommand{\THEN}{\mbox{\mathem then}}

\newcommand{\WHERE}{\mbox{\mathem where}}

\newcommand{\canonic}[1]{{\cal C}(#1)}
\newcommand{\inj}{inj}

\newcommand{\arrow}{\rightarrow}

\newcommand{\comment}[1]{}
\newcommand{\ignore}[1]{}
\newcommand{\pjs}[1]{}
\newcommand{\jw}[1]{}
\newcommand{\lang}{{\cal L}}

\newcommand{\deriv}[2]{d_{#2}(#1)} 
\newcommand{\derivC}[2]{\canonic{d_{#2}(#1)}} 

\newcommand{\mkEmp}[1]{{\mathit mkEmpty}_{#1}} 

\newcommand{\RExp}{\textit{RE}}

\newcommand\semeq\equiv              
\newcommand{\conc}{\cdot}
\newcommand{\simp}[1]{\mathit{cnf}(#1)}

\newcommand\Descendant\preceq

\newcommand\power{\wp}
\newcommand{\shuffle}{\|}

\newcommand{\flatten}[1]{|#1|}

\newcommand{\SeqOp}{\textsc{Seq}}
\newcommand\Seq[2]{\textsc{Seq}\ #1\ #2}
\newcommand{\Left}{\textsc{Inl}}
\newcommand{\Right}{\textsc{Inr}}
\newcommand{\Fold}{\textsc{Fold}}
\newcommand{\Eps}{\textsc{Eps}}
\newcommand{\Sym}{\textsc{Sym}}

\newcommand{\Kons}{{\cal K}}

\newcommand{\Desc}[1]{{\cal D}(#1)}

\title{Solving of Regular Equations Revisited (extended version)}

\author{Martin Sulzmann\inst{1} and Kenny Zhuo Ming Lu \inst{2}}
\institute{Karlsruhe University of Applied Sciences \\
  \email{martin.sulzmann@hs-karlsruhe.de}
  \and
   Nanyang Polytechnic \\
  \email{luzhuomi@gmail.com}}

\begin{document}

\maketitle

\begin{abstract}
  Solving of regular equations via Arden's Lemma is folklore knowledge.
  We first give a concise algorithmic specification of all elementary solving steps.
  We then discuss  a computational interpretation
  of solving in terms of coercions that transform parse trees of regular equations
  into parse trees of solutions.  
  Thus, we can identify some conditions on the shape of regular equations
  under which resulting solutions are unambiguous.
  We apply our result to convert a DFA to an unambiguous regular expression.
  In addition, we show that operations such as subtraction and shuffling
  can be expressed via some appropriate set of regular equations.
  Thus, we obtain direct (algebraic) methods without having
  to convert to and from finite automaton.
    \\ \mbox{} 
\end{abstract}

\section{Introduction}

The conversion of a regular expression (RE) into a deterministic finite automaton (DFA)
is a well-studied topic. Various methods and optimized implementations exist.
The opposite direction has received less attention.
In the literature, there are two well-known methods to
translate DFAs to REs, namely,
state elimination~\cite{DBLP:journals/tc/BrzozowskiM63} and
solving of equations via Arden's Lemma~\cite{DBLP:conf/focs/Arden61}.

The solving method works by algebraic manipulation of equations.
Identity laws are applied to change the syntactic form of an equation's right-hand side
such that Arden's Lemma is applicable. Thus, the set of equations is reduced
and in a finite number of steps a solution can be obtained.
State elimination has a more operational flavor and reduces states by introducing transitions labeled with regular expressions.
The state elimination method appears to be better studied in the literature.
For example, see the works~\cite{Han:2013:SEH:2595015.2595018,DBLP:journals/corr/abs-1008-1656,Ahn:2009:ISE:1577697.1577720}
that discuss heuristics to obtain short regular expressions.

In this paper, we revisit solving of regular equations via Arden's Lemma.
Specifically, we make the following contributions:

\begin{itemize}
\item We give a concise algorithmic description of solving of regular equations
  where we give a precise specification of all algebraic laws applied (Section~\ref{sec:solve-reg-eq}).  

\item We give a computational interpretation of solving by means of coercions
      that transform parses tree of regular equations
      into parse trees of solutions.
      We can identify simple criteria on the shape of regular equations
      under which resulting solutions are unambiguous (Section~\ref{sec:comput-inter}).

\item We apply our results to the following scenarios:
  \begin{itemize}
  \item We show that regular expressions obtained from DFAs via Brzozowski's algebraic method
    are always unambiguous (Section~\ref{sec:algebraic-method}).
    
  \item We provide direct, algebraic methods to obtain the subtraction 
        and shuffle among two regular expressions (Section~\ref{sec:subtract-intersect}).
        Correctness follows via some simple coalgebraic reasoning.
  \end{itemize}
\end{itemize}

We conclude in Section~\ref{sec:related-works}
where we also discuss related works.

The appendix contains further details such as proofs and a parser
for regular equations. We also report on an implementation
for solving of regular equations in Haskell~\cite{regex-symbolic}
including benchmark results.

\section{Preliminaries}

Let $\Sigma$ be a finite set of symbols (literals) with   $x$, $y$, and $z$
ranging over $\Sigma$.
We write $\Sigma^*$ for the set of finite words over $\Sigma$, 
$\varepsilon$ for the empty word, and
$v \conc w$ for the concatenation of words $v$ and $w$. A language is
a subset of $\Sigma^*$.

\begin{defn}[Regular Languages]
  The set ${\cal R}$ of regular languages is defined inductively
  over some alphabet $\Sigma$ by
  \bda{lcl}
  R, S & ::= & \emptyset \mid \{ \varepsilon \} \mid \{ x \} \mid (R + S) \mid (R \conc S) \mid (R^*)
  \ \ \ \ \ \mbox{where $x \in \Sigma$.}
  \eda
  Each regular language is a subset of $\Sigma^*$
  where we assume that
   $R\conc S$ denotes $\{ v \conc w \mid v \in R \wedge w \in S\}$,
  $R+S$ denotes $R \cup S$
  and $R^*$ denotes $\{ w_1 \conc \cdots \conc w_n \mid n \ge0 \wedge \forall i \in \{1,\ldots,n\}.~ w_i \in R \}$.

  We write $R \semeq S$ if $R$ and $S$ denote the same set of words.
\end{defn}

We often omit parentheses by assuming that $^*$ binds tighter than $\conc$ and $\conc$ binds
tighter than $+$.
As it is common, we assume that $+$ and $\conc$ are right-associative.
That is, $R + S + T$ stands for $(R + (S + T))$ and
$R + S + R\conc S \conc T$ stands for $R + (S + (R \conc (S \conc T)))$.

%

\begin{defn}[Regular Expressions]
  The set $\RExp$ of \emph{regular expressions} is defined inductively
  over some alphabet $\Sigma$ by
\bda{lcl}
 r,s & ::= & \phi \mid \varepsilon \mid x \mid (r + s) \mid (r
 \conc s) \mid (r^*) \ \ \ \ \ \mbox{where $x \in \Sigma$.}
 \eda
\end{defn}
%

\begin{defn}[From Regular Expressions to Languages]
 The meaning function $\lang$ maps a regular expression to a language.
It is defined inductively as follows:
 \\
 $\lang(\phi) = \{ \}$.
 $\lang(\varepsilon) = \{ \varepsilon \}$.
 $\lang(x) = \{ x \}$.
 $\lang(r + s) = (\lang(r) + \lang(s))$.
 $\lang(r \conc s) = (\lang(r) \conc \lang(s))$.
 $\lang(r^*) = (\lang(r)^*)$.
\end{defn}
We say that regular expressions $r$ and $s$ are equivalent, $r\semeq s$, if $\lang(r) = \lang(s)$.
\begin{defn}[Nullability]
A regular expression $r$ is \emph{nullable} if $\varepsilon \in \lang(r)$.
\end{defn}

\begin{lem}[Arden's Lemma \cite{DBLP:conf/focs/Arden61}]  
  Let $R$, $S$, $T$ be regular languages where $\varepsilon \not\in S$.
  Then, we have that
  $R \semeq S \conc R + T$ iff $R \semeq S^* \conc T$.
\end{lem}
The direction from right to left holds in general.
For the direction from left to right, pre-condition $\varepsilon \not\in S$ is required.
For our purposes, we only require the direction from right to left.

\section{Solving Regular Equations}
\label{sec:solve-reg-eq}

\begin{defn}[Regular Equations]
\label{def:regular-equations}  
  We write $E$ to denote a \emph{regular equation}
  of the form $R \regEq \alpha$ where the form of the right-hand side $\alpha$ is as follows.
  \bda{lcll}
  \alpha && ::= & r \conc R \mid r \mid \alpha + \alpha
  \eda
  In addition to $\alpha$, we will sometimes use $\beta$ to denote right-hand sides.
  
  We will treat regular language symbols $R$ like variables.
  We write $r,s,t$ to denote expressions that do not refer to symbols $R$.

  We write $R \in \alpha$ to denote that $R$ appears in $\alpha$.
  Otherwise, we write $R \not\in \alpha$.
  
  We write ${\cal E}$ to denote a set $\{R_1 \regEq \alpha_1, \dots, R_n \regEq \alpha_n\}$ of regular equations.
  We assume that (1) left-hand sides are distinct by requiring that $R_i \not= R_j$ for $i \not= j$,
  and (2) regular language symbols on right-hand sides appear on some left-hand side
  by requiring that for any $R \in \alpha_j$ for some $j$ there exists $i$ such that $R = R_i$.
  We define ${\it dom}({\cal E}) = \{ R_1, \dots, R_n\}$.

\end{defn}

Regular languages are closed under union and concatenation, hence, we can guarantee
the existence of solutions of these variables in terms of regular expressions.

\begin{defn}[Solutions]
  We write $\{ R_1 \mapsto \gamma_1, \dots, R_n \mapsto \gamma_n \}$
  to denote an idempotent substitution mapping $R_i$ to $\gamma_i$
  where $\gamma_i$ denote expressions that may consist of a mix of regular expressions
  and regular language symbols $R$.

  Let $\psi = \{ R_1 \mapsto \gamma_1, \dots, R_n \mapsto \gamma_n \}$ be a substitution
  and $\gamma$ some expression.
  Then, $\psi(\gamma)$ is derived from $\gamma$ by replacing each occurrence of $R_i$ by $\gamma_i$.

  Let ${\cal E} = \{ R_1 \regEq \alpha_1, \dots, R_n \regEq \alpha_n \}$.
  Then, we say that $\psi$ is a \emph{solution} for ${\cal E}$ if
  $\psi(R_i)$, $\psi(\alpha_i)$ are regular expressions where
  $\psi(R_i) \semeq \psi(\alpha_i)$ for $i=1,\dots,n$.
\end{defn}

We solve equations as follows.
We apply Arden's Lemma on equations that are of a certain (normal) form
$R \regEq s \conc R + \alpha$
where $R \not\in \alpha$.
Thus, we can eliminate this equation
by substituting $R$ with $s^* \conc \alpha$ on all right-hand sides.
In case of $R \regEq \alpha$ where $R \not\in \alpha$ we can substitute directly.
We repeat this process until all equations are solved.
Below, we formalize the technical details.

\begin{defn}[Normal Form]
    We say that $R \regEq \alpha$ 
    is in \emph{normal form} iff either (1) $R \not\in \alpha$, or (2)
    $\alpha = s_1 \conc R_1 + \dots + s_n \conc R_n + t$ such that
    $R = R_1$ and $R_i \not=R_j$ for $i \not=j$.
\end{defn}
Recall that $t$ does not refer to symbols $R$.
Every equation can be brought into normal form by applying
the following algebraic equivalence laws.

\begin{defn}[Equivalence]
  We say two expressions $\gamma_1$ and $\gamma_2$ are \emph{equivalent},
    written $\gamma_1 \simeq \gamma_2$, if one  can be transformed into the
    other by application of the following rules.
    \bda{c}
    \tlabel{E1} \ \ \gamma_1 \conc (\gamma_2 + \gamma_3) \simeq \gamma_1 \conc \gamma_2 + \gamma_1 \conc \gamma_3
    \ \ \ \
    \tlabel{E2} \ \ \gamma_1 \conc (\gamma_2 \conc \gamma_3) \simeq (\gamma_1 \conc \gamma_2) \conc \gamma_3
    \\ \\
    \tlabel{E3} \ \ \gamma_1 + (\gamma_2 + \gamma_3) \simeq (\gamma_1 + \gamma_2) + \gamma_3
    \ \ \ \
    \tlabel{E4} \ \ \gamma_2 \conc \gamma_1 + \gamma_3 \conc \gamma_1 \simeq (\gamma_2 + \gamma_3) \conc \gamma_1
    \\
    \tlabel{E5} \ \ \gamma_1 + \gamma_2 \simeq \gamma_2 + \gamma_1
    \ \ \ \
    \tlabel{E6} \ \
    \myirule{\gamma_1 \simeq \gamma_2}
            {\beta[\gamma_1] \simeq \beta[\gamma_2]}
    \ \ \ \
    \tlabel{E7} \ \
    \myirule{\gamma_1 \simeq \gamma_2 \ \ \gamma_2 \simeq \gamma_3}
            {\gamma_1 \simeq \gamma_3}
   \eda
   Rule \tlabel{E6} assumes expressions with a hole.
    \bda{c}
      \beta[] ::= [] \mid \beta[] + \beta \mid \beta + \beta[]
    \eda
    We write $\beta[\gamma]$ to denote the expression where the hole $[]$
    is replaced by $\gamma$. 
\end{defn}

We formulate solving of equations in terms of a rewrite system
among a configuration $\Config{\psi}{\EE}$
where substitution $\psi$ represents the so far accumulated solution
and $\EE$ the yet to be solved set of equations.

\begin{defn}[Solving]
\label{def:solving}  
Let ${\cal E} = \{R_1 \regEq \alpha_1, \dots, R_n \regEq \alpha_n \}$.
Then, we write $R \regEq \alpha \uplus {\cal E}'$
to denote the set that equals to ${\cal E}$
where $R \regEq \alpha$ refers to some equation in ${\cal E}$
and ${\cal E}'$ refers to the set of remaining equations.
  
  \bda{c}
  \tlabel{Arden} \ \ \ \
  \myirule{R \not\in \alpha}
     {
   \solveRel{\Config{\psi}{ R \regEq s \conc R + \alpha \uplus {\cal E}}}
            {\Config{\psi}{ R \regEq s^* \conc \alpha \uplus {\cal E}}} }
  \\ \\
  \tlabel{Subst} \ \ \ \          
  \myirule{R \not\in \alpha
           \\ \psi'  =  \{ R \mapsto \alpha \} \cup \{ S \mapsto \{ R \mapsto \alpha \}(\gamma) \mid S \mapsto \gamma \in \psi \}
           \\ \EE'  =  \{ R' \regEq \alpha'' \mid R' \regEq \alpha' \in \EE \wedge
                                                     \{ R \mapsto \alpha \}(\alpha') \simeq \alpha'' \}
           }
          {\solveRel{\Config{\psi}{R \regEq \alpha \uplus {\cal E}}}
                    {\Config{\psi'}{\EE'}}
          }
    \eda
    We write $\Rightarrow^*$ to denote the transitive and reflexive closure of solving steps $\Rightarrow$.
\end{defn}
Initially, all equations are in normal form.
Rule \tlabel{Arden} applies Arden's Lemma on some equation in normal form.
Rule \tlabel{Subst} removes an equation
$R \regEq \alpha$ where $R \not\in \alpha$.
The substitution $\{ R \mapsto \alpha \}$ implied by the equation is applied
on all remaining right-hand sides.
To retain the normal form property of equations,
we normalize right-hand sides by applying rules \tlabel{E1-7}.
The details of normalization are described in the proof of the upcoming statement.
We then extend the solution
accumulated so far by adding $\{ R \mapsto \alpha \}$.
As we assume substitutions are idempotent,
$\{ R \mapsto \alpha \}$ is applied on all expressions
in the codomain of $\psi$.

\begin{thm}[Regular Equation Solutions]
  Let ${\cal E}$ be a set of regular equations in normal form.
  Then, $\solveRels{\Config{\{ \}}{{\cal E}}}{\Config{\psi}{\{ \}}}$ for some substitution $\psi$
    where $\psi$ is a solution for ${\cal E}$.
\end{thm}
\begin{proof}
  We first observe that rule \tlabel{Arden}
  and \tlabel{Subst} maintain the normal form property for equations.
  This immediately applies to rule \tlabel{Arden}.

  Consider rule \tlabel{Subst}.
  Consider $R' \regEq \alpha'$.
  We need to show that $\{ R \mapsto \alpha \}(\alpha')$
  can be transformed to some form $\alpha''$
  such that $R' \regEq \alpha''$ is in normal form.

    If $R \not\in \alpha'$ nothing needs to be done as we assume that equations are initially in normal form.
    
  Otherwise, 
  we consider the possible shapes of $\alpha$ and $\alpha'$.
  W.l.o.g.~$\alpha'$ is of the form
  $t_1 \conc R_1 + \dots + r \conc R + \dots + t_n \conc R_n + t'$
  and $\alpha$ is of the form
  $s_1 \conc T_1 + \dots  + s_k \conc T_k + t''$.
  We rely here on rule \tlabel{E3} that allows us to drop parentheses
  among summands.
  
  $R$ is replaced by $\alpha$ in $\alpha'$.
  This generates the subterm
  $r \conc (s_1 \conc T_1 + \dots  + s_k \conc T_k + t'')$.
  On this subterm, we exhaustively apply rules \tlabel{E1-2}.
  This yields the subterm
  $(r \conc s_1) \conc T_1 + \dots + (r \conc s_k) \conc T_k + t''$.

  This subterm is one of the sums in the term obtained
  from $\{ R \mapsto \alpha \}(\alpha')$.
  Via rules \tlabel{E6-7} the above transformation steps
  can be applied on the entire term $\{ R \mapsto \alpha \}(\alpha')$.
  Hence, this term can be brought into the form
  $r_1 \conc S_1 + \dots + r_m \conc S_m + t$.
  Subterm $t$ equals $t' + t''$ and
  subterms $r_i \conc S_i$ refer to one of the subterms
  $t_j \conc R_j$ or $(r \conc s_l) \conc T_l$.

  We are not done yet because
  subterms $r_i \conc S_i$ may contain duplicate symbols.
  That is, $S_i = S_j$ for $i \not=j$.
  We apply rule \tlabel{E4} in combination with rule \tlabel{E3} and \tlabel{E4}
  to combine subterms with the same symbol.
  Thus, we reach the form
  $r_1' \conc R_1' + \dots + r_o' \conc R_o' + t$ such that $R_i \not=R_j$ for $i \not=j$.

  If $R' \not= R'_i$ for $i=1,\dots,o$ we are done.
  Otherwise, $R = R'_i$ for some $i$.
  We apply again \tlabel{E3} and \tlabel{E5}
  to ensure that the component $s_i \conc R_i$
  appears first in the sum.

  Next, we show that within a finite number of \tlabel{Arden}
  and \tlabel{Subst} rule applications we reach the
  configuration $\Config{\psi}{\{ \}}$.
  For this purpose, we define an ordering relation
  among configurations $\Config{\psi}{\EE}$.
  
  For $\EE = \{R_1 \regEq \alpha_1, \dots, R_n \regEq \alpha_n \}$
  we define
  $$\varsEE{\EE} = (\{R_1,\dots,R_n\}, \multiset{S_1,\dots,S_m})
  $$
  where $\multiset{\dots}$ denotes a multi-set and
  $S_j$ are the distinct occurrences of symbols appearing
  on some right-hand side $\alpha_i$.
  Recall that by construction $\{S_1,\dots,S_m\} \subseteq \{R_1,\dots,R_n\}$.
  See (2) in Definition~\ref{def:regular-equations}.
  We define $\Config{\psi}{\EE} < \Config{\psi'}{\EE'}$
  iff either
  (a) $M \subsetneq M'$
  or
  (b) $M = M'$ and the number of symbols in $N$ is strictly smaller
  than the number of symbols in $N'$
  where $\varsEE{\EE} = (M,N)$ and $\varsEE{\EE'} = (M',N')$.

  For sets $\EE$ of regular equations as defined in Definition~\ref{def:regular-equations} this is a well-founded order.
  Each of the rules \tlabel{Subst}
  and \tlabel{Arden} yield a smaller configuration w.r.t~this order.
  For rule \tlabel{Subst} case (a) applies whereas for rule \tlabel{Arden}
  case (b) applies.
  Configuration $\Config{\psi}{\{\}}$ for some $\psi$ is the minimal element.
  Hence, in a finite number of rule applications
  we reach $\Config{\psi}{\{\}}$.
  
  Substitution $\psi$ must be a solution because
  (1) normalization steps   are equivalence preserving
  and (2) based on Arden's Lemma
  we have that every solution for $R \regEq s^* \conc \alpha$
  is also a solution for $R \regEq s \conc R + \alpha$.
  \qed
\end{proof}

\begin{example}
\label{ex:solve1}
Consider
${\cal E} = \{ R_1 \regEq x \conc R_1 + y \conc R_2 + \varepsilon,
R_2 \regEq y \conc R_1 + x \conc R_2 + \varepsilon \}$.
For convenience, we additionally make use of  associativity of concatenation ($\conc$).
\bda{ll}
& \Config{\{ \} }
         {\{ R_1 \regEq x \conc R_1 + y \conc R_2 + \varepsilon,
             R_2 \regEq y \conc R_1 + x \conc R_2 + \varepsilon \}}
\\
\stackrel{\tlabel{Arden}}{\Rightarrow} &
\Config{\{ \} }
         {\{ R_1 \regEq x^* \conc (y \conc R_2 + \varepsilon),
           R_2 \regEq y \conc R_1 + x \conc R_2 + \varepsilon \}}
\\
\stackrel{\tlabel{Subst}}{\Rightarrow} &
( y \conc (x^* \conc (y \conc R_2 + \varepsilon)) + x \conc R_2 + \varepsilon \simeq
(y \conc x^* \conc y + x) \conc R_2 + y \conc x^* \conc \varepsilon + \varepsilon )
\\ &
\Config{\{ R_1 \mapsto x^* \conc (y \conc R_2 + \varepsilon) \}}
       {\{ R_2 \regEq (y \conc x^* \conc y + x) \conc R_2 + y \conc x^* \conc \varepsilon + \varepsilon  \}}
\\
\stackrel{\tlabel{Arden}}{\Rightarrow} &
\Config{\{ R_1 \mapsto x^* \conc (y \conc R_2 + \varepsilon) \}}
       {\{ R_2 \regEq (y \conc x^* \conc y + x)^* \conc (y \conc x^* \conc \varepsilon + \varepsilon)  \}}
\\
\stackrel{\tlabel{Subst}}{\Rightarrow} &       
\Config{\{ R_1 \mapsto x^* \conc (y \conc (y \conc x^* \conc y + x)^* \conc (y \conc x^* \conc \varepsilon + \varepsilon) + \varepsilon),
         \\ & R_2 \mapsto (y \conc x^* \conc y + x)^* \conc (y \conc x^* \conc \varepsilon + \varepsilon) \} }
       {\{   \}}
\eda

The formulation in Definition~\ref{def:solving}
leaves the exact order in which equations are solved unspecified.
Semantically, this form of non-determinism has no impact
on the solution obtained.
However, the syntactic shape of solutions is sensitive
to the order in which equations are solved.

Suppose we favor the second equation which then yields
the following.
\bda{ll}
& \Config{\{ \} }
         {\{ R_1 \regEq x \conc R_1 + y \conc R_2 + \varepsilon,
             R_2 \regEq y \conc R_1 + x \conc R_2 + \varepsilon \}}
\\
\Rightarrow^* 
& \Config{\{ R_1 \mapsto (x + y \conc x^* \conc y)^* + y \conc x^* + \varepsilon,
       \\ & R_2 \mapsto x^* \conc (y \conc ((x + y \conc x^* \conc y)^* + y \conc x^* + \varepsilon) + \varepsilon)\}}
         {\{ \}}
\eda
where for convenience, we exploit the law $r \conc \varepsilon \semeq r$.
\end{example}


\section{Computational Interpretation}
\label{sec:comput-inter}

We characterize under which conditions solutions to regular equations are unambiguous.
By unambiguous solutions we mean that the resulting expressions are unambiguous.
An  expression is ambiguous if there exists
a word which can be matched in more than one way.
That is, there must be two distinct parse trees
which share the same underlying word~\cite{DBLP:conf/ppdp/BrabrandT10}.

We proceed by establishing the notion of a parse tree.
Parse trees capture the word that has been matched and also record
which parts of the regular expression have been matched.
We follow~\cite{cduce-icalp04} and view expressions as types and parse trees as values.

\begin{defn}[Parse Trees]
\label{def:parse-trees}  
\bda{c}
 u,v \ ::= \ \Eps \mid \Sym\ x \mid \Seq v v
 \mid  \Left~v \mid \Right~v
    \mid  vs \mid \Fold\ v
 \ \ \ \ 
vs \ ::= \ [] \mid v : vs
\eda
The valid relations among parse trees and regular expressions
are defined via a natural deduction style
proof system.
\bda{c}
\EE \turns [] : r^*
 \ \ \ \        
 \EE \turns \Eps : \epsilon
 \ \ \ \
 \myirule{x \in \Sigma}
         { \EE \turns \Sym\ x : x }        
\\
\myirule{\EE \turns v : r
         \sgap \EE \turns vs : r^*}
        {\EE \turns (v:vs) : r^*}
 \ \ \ \
\myirule{ \EE \turns v_1 : r_1 \sgap \EE \turns v_2 : r_2 }
        {\EE \turns \Seq{v_1}{v_2} : r_1 \conc r_2 }
\\ 
 \myirule{\EE \turns v_1 : r_1}
         {\EE \turns \Left~v_1 : r_1 + r_2}
 \ \ \ \
 \myirule{\EE \turns v_2 : r_2}
         {\EE \turns \Right~v_2 : r_1 + r_2}
\ \ \ \
  \myirule{ \EE \turns v : \alpha \ \ \ \ R \regEq \alpha \in \EE }
         { \EE \turns \Fold \ v : R}         
\eda
For expressions not referring to variables
we write $\turns v : r$ as a shorthand for $\{ \} \turns v : r$.
\end{defn}
Parse tree values are built using data constructors. 
The constant constructor $\Eps$ represents the value
belonging to the empty word regular expression.
For letters, we use the unary constructor $\Sym$ to record the symbol.
In case of choice ($+$), we use unary constructors $\Left$ and $\Right$
to indicate if either the left or right expression is part of the match.
For repetition (Kleene star) we use Haskell style lists where  
we write $[v_1,...,v_n]$ as a short-hand
for the list $v_1 : ... : v_n : []$.
In addition to the earlier work~\cite{cduce-icalp04},
we introduce a  $\Fold$  constructor and a proof rule to (un)fold a regular equation.

\begin{example}
\label{ex:fold-parse-tree}  
  Consider $\EE = \{ R \regEq  x \conc R + y \}$.
  Then, we find that
  $$\EE \turns \Fold\ (\Left\ (\Seq\ (\Sym\ x) \ (\Fold\ (\Right\ (\Sym\ y))))) : R
  $$
  The equation is unfolded twice
  where we first match against the left part $x \conc R$
  and then against the right part $y$.
\end{example}  

The relation established in Definition~\ref{def:parse-trees}
among parse trees, expressions and equations is correct
in the sense that (1) flattening of the parse tree yields
a word in the language and (2) for each word there exists
a parse tree.

\begin{defn}[Flattening]
  We can flatten a parse tree to a word as follows:
\bda{lllllllllllllllll}
 \flatten{\Eps}  &=&  \epsilon &
  \flatten{\Sym\ x} &=& x          &
  \flatten{\Left~v}  &=&  \flatten{v} ~~~~&
  \flatten{v:vs}  &=&  \flatten{v} \conc \flatten{vs}
\\
 \flatten{[]} & =&  \epsilon  ~~~&
 \flatten{\Seq{v_1}{v_2}} &=& \flatten{v_1} \conc \flatten{v_2} ~~~&
 \flatten{\Right~v} &=&  \flatten{v}  &
 \flatten{\Fold\ v} &=&  \flatten{v}
\eda
\end{defn}

\begin{prop}
  Let $\EE$ be a set of regular equations and $\psi$ a solution.
  Let $R \in {\it dom}(\EE)$.
  (1) If $w \in \lang(\psi(R))$ then $\EE \turns v : R$
  for some parse tree $v$ such that $\flatten{v} = w$.
  (2) If $\EE \turns v : R$ then $\flatten{v} \in \lang(\psi(R))$.
\end{prop}

The above result follows by providing a parser for regular equations.
For (1) it suffices to
compute a parse tree if one exists.
For (2) we need to enumerate all possible parse trees.
This is possible by extending our prior work~\cite{DBLP:conf/flops/SulzmannL14,DBLP:journals/ijfcs/SulzmannL17}
to the regular equation setting.
Details are given in Appendix~\ref{sec:parse-reg-eq}.

Parse trees may not be unique because some equations/expressions
may be ambiguous in the sense that a word can be matched in more than one way.
This means that there are two distinct parse trees representing the same word.
We extend the notion of ambiguous expressions~\cite{DBLP:conf/ppdp/BrabrandT10}
to the setting of regular equations.

\begin{defn}[Ambiguity]
Let $\EE$ be a set of regular equations and $r$ be an expression.  
We say $r$ is \emph{ambiguous} w.r.t.~$\EE$ iff
there exist two distinct parse trees $v_1$ and $v_2$ such
that $\EE \turns v_1 : r$ and $\EE \turns v_2 : r$ 
where $\flatten{v_1} = \flatten{v_2}$.
\end{defn}

\begin{example}
  $[\Left \ (\Seq{(\Sym\ x)}{(\Sym\ y)})]$ and $[\Right \ (\Left\ (\Sym\ x)), \Right \ (\Right\ (\Sym\ y))]$
are two distinct parse trees for expression $(x\conc y + x + y)^*$ (where $\EE = \{ \}$)
and word $x \conc y$.
\end{example}

On the other hand, the equation from Example~\ref{ex:fold-parse-tree} is unambiguous
due to the following result.

\begin{defn}[Non-Overlapping Equations]
  We say an equation $E$ is \emph{non-overlapping} if $E$ is
  of the following form $R \regEq x_1 \conc R_1 + \dots + x_n \conc R_n + t$
  where $x_i \not= x_j$ for $i \not= j$ and either $t=\varepsilon$ or $t = \phi$.
\end{defn}

Equation $R \regEq  x \conc R + y$ does not exactly match the above definition.
However, we can transform  $\EE = \{ R \regEq  x \conc R + y \}$
into the equivalent set $\EE' = \{ R \regEq  x \conc R + y \conc S, S \regEq \varepsilon \}$
that satisfies the non-overlapping condition.

\begin{prop}[Unambiguous Regular Equations]
\label{prop:regeq-unambig}  
  Let ${\cal E}$ be a set of non-overlapping equations where $R \in {\it dom}({\cal E})$.
  Then, we have that $R$ is unambiguous.
\end{prop}

Ultimately, we are interested in obtaining a parse tree for the resulting solutions
rather than the original set of equations.
For instance, the solution for Example~\ref{ex:fold-parse-tree} is $x^* \conc y$.
Hence, we wish to transform the parse tree
$$\Fold\ (\Left\ (\Seq\ (\Sym\ x) \ (\Fold\ (\Right\ (\Sym\ y)))))
$$
into a parse tree for $x^* \conc y$.
Furthermore, we wish to guarantee that if equations are unambiguous so are solutions.
We achieve both results by explaining each solving step among regular equations
in terms of a (bijective) transformation among the associated parse trees.

We refer to these transformations as coercions as they operate on parse trees.
We assume the following term language to represent coercions.

\begin{defn}[Coercion Terms]
  Coercion terms $c$ and
  patterns $pat$ are inductively defined by
  \bda{lrl}
  c & ::= & v \mid k \mid \lambda v.c \mid c \ c \mid \REC\ x.c
  \mid \CASE\ c \ \OF\ [pat_1 \Rightarrow c_1, \ldots, pat_n \Rightarrow c_n]
  \\
  pat & ::= & y \mid k \ pat_1 \ ... pat_{arity(k)}
  \eda
  where pattern variables $y$ range overs a denumerable set of variables disjoint from $\Sigma$
  and constructors $k$ are taken from the set
  $\Kons = \{ \Eps, \SeqOp, \Left, \Right, \Fold \}$.
  The function~$arity(k)$ defines the arity of constructor $k$.
  Patterns are linear (i.e., all pattern variables are distinct) and
  we write $\lambda pat.c$ as a shorthand for
  $\lambda v. \CASE\ v \ \OF\ [pat \Rightarrow c]$.
\end{defn}

We give meaning to coercions in terms of a standard big-step operational semantics.
Given a coercion (function) $f$ and some (parse tree) value $u$,
we write $f \ u \Downarrow v$ to denote the evaluation of $f$ for input $u$
with resulting (parse tree) value $v$.
We often write $f (u)$ as a shorthand for $v$.
We say a coercion $f$ is \emph{bijective} if there exists a coercion $g$
such that for every $u, v$ where $f \ u \Downarrow v$ we have
that $g \ v \Downarrow u$. We refer to $g$ as the \emph{inverse} of $f$.

We examine the three elementary solving steps, Arden, normalization and substitution.
For each solving step we introduce an appropriate (bijective) coercion to carry out the transformation
among parse trees.

\begin{lem}[Arden Coercion]
\label{le:arden-coercion}  
  Let $\EE$ be a set of regular equations where $R \regEq s \conc R + \alpha \in \EE$
  such that $R \not\in \alpha$
  and $\EE \turns \Fold\ v : R$ for some parse tree $v$.
  Then, there exists a bijective coercion $f_A$ such that
  $\EE \turns f_A (v) : s^* \conc \alpha$ where $\flatten{v} = \flatten{f_A (v)}$.
\end{lem}
\begin{proof}
  By assumption $\EE \turns v : s \conc R + \alpha$.
  The following function $f_A$ satisfies $\EE \turns f_A (v) : s^* \conc \alpha$ where $\flatten{v} = \flatten{f_A (v)}$.
  For convenience we use symbols $v$ and $u$ as pattern variables.
\bda{llll}
f_A = & \REC\ f. \lambda x. & \CASE\ x \ \OF\
\\   && [ \Right\ u \Rightarrow \Seq{[]}{u},
  \\ &&    \Left\ (\Seq{u}{(\Fold\ v)} \Rightarrow & \CASE\ f(v) \ \OF
  \\ &&&          [\Seq{us}{u_2} \Rightarrow \Seq{(u:us)}{u_2}] ]
\eda
Function $f_A$ is bijective. Here is the inverse function.
\bda{llll}
f_A^{-1} = & \REC\ g. \lambda x. & \CASE\ x \ \OF
\\ && [\Seq{[]}{u} \Rightarrow \Fold\ (\Right\ u),
\\ &&  \Seq{(v:vs)}{u} \Rightarrow \Fold\ (\Left \ (\Seq{v}{(g \ (\Seq{vs}{u}))}))]
\eda
\qed
\end{proof}

\begin{lem}[Normalization Coercion]
\label{le:normalize-coercion}  
  Let $\gamma_1, \gamma_2$ be two expressions such that $\gamma_1 \simeq \gamma_2$
  and $\EE \turns v : \gamma_1$ for some set $\EE$ and parse tree $v$.
  Then, there exists a bijective coercion f such that
  $\EE \turns f (v) : \gamma_2$ where $\flatten{v} = \flatten{f (v)}$.
\end{lem}
\begin{proof}
  For each of the equivalence proof rules, we introduce an appropriate (bijective) coercion.
  For rule \tlabel{E1} we employ
  \bda{ll}
  f_{E_1} = \lambda v. & \CASE\ v \ \OF
    \\ & [\Seq{u}{(\Left\ v)} \Rightarrow \Left\ (\Seq{u}{v}),
    \\ &  \Seq{u}{(\Right\ v)} \Rightarrow \Right\ (\Seq{u}{v})]
  \eda
  where the inverse function is as follows.
  \bda{ll}
  f_{E_1}^{-1} = \lambda v. & \CASE\ v \ \OF
    \\ & [\Left\ (\Seq{u}{v}) \Rightarrow \Seq{(\Left\ u)}{v},
    \\ &  \Right\ (\Seq{u}{v}) \Rightarrow \Seq{(\Right\ u)}{v}]
  \eda  
  Coercions for rules \tlabel{E2-5} can be defined similarly.
  Rule \tlabel{E7} corresponds to function composition and
  rule \tlabel{E6} requires to navigate to the respective hole position.
  Details are omitted for brevity.
  \qed
\end{proof}

We will write $\gamma_1 \stackrel{f}{\simeq} \gamma_2$ to denote the coercion $f$
to carry out the transformation of $\gamma_1$'s parse tree into $\gamma_2$'s parse tree.

What remains is to define coercions to carry out substitution where we replace subterms.

\begin{defn}[Substitution Context]
 We define expressions with multiple holes to characterize substitution of a subterm by another.
  \bda{lcl}
  \delta\Angle{} ::= r \conc \Angle{}
  \mid  \delta\Angle{} + \delta\Angle{}
  \mid \delta\Angle{} + \alpha
  \mid \alpha + \delta\Angle{}
  \eda
  We refer to $\delta\Angle{}$ as a \emph{substitution context}. 

We define a set of functions indexed by the shape of a substitution context.
For $\delta\Angle{}$ we transform
$\alpha\Angle{R}$'s parse tree into $\alpha\Angle{\alpha}$'s parse tree
assuming the equation $R \regEq \alpha$.
\bda{ll}
f_{r \conc \Angle{}} & =
\ba{ll}
\lambda u. & \CASE\ u \ \OF
\\ & [\Seq{u}{(\Fold\ v)} \Rightarrow \Seq{u}{v}]
\ea
\\ \\
f_{\delta\Angle{} + \delta\Angle{}} & =
\ba{ll}
\lambda u. & \CASE\ u \ \OF
\\ & [\Left\ v \Rightarrow \Left\ (f_{\delta\Angle{}} (v)),
  \\ & \Right\ v \Rightarrow \Right\ (f_{\delta\Angle{}} (v))]
\ea
\eda
\bda{ccc}
\ba{ll}
f_{\delta\Angle{} + \alpha} & =
\ba{ll}
\lambda u. & \CASE\ u \ \OF
\\ & [\Left\ v \Rightarrow \Left\ (f_{\delta\Angle{}} (v)),
  \\ & \Right\ v \Rightarrow  \Right\ v]
\ea
\ea
& \ \ \ \ &
\ba{ll}
f_{\alpha + \delta\Angle{}} & =
\ba{ll}
\lambda u. & \CASE\ u \ \OF
\\ & [\Left\ v \Rightarrow \Left\ v,
\\ & \Right\ v \Rightarrow \Right\ (f_{\delta\Angle{}} (v))]
\ea
\ea
\eda
\end{defn}
Functions $f_{\delta\Angle{}}$ navigate to the to-be-replaced subterm and
drop the $\Fold$ constructor if necessary.
There are inverse functions which we omit for brevity.

\begin{lem}[Substitution Coercion]
\label{le:subst-coercion}  
  Let $\EE$ be a set of equations, $R \regEq \alpha \in \EE$
  such that $\EE \turns v : \delta\Angle{R}$ for some parse tree $v$ and
  substitution context $\delta\Angle{}$.
  Then, we find that $\EE \turns f_{\delta\Angle{}} (v) : \delta\Angle{\alpha}$
  where $\flatten{v} = \flatten{f_{\delta\Angle{}} (v)}$.
\end{lem}
\begin{proof}
  Follows by induction over the structure of $\delta\Angle{}$.
  \qed
\end{proof}

We integrate the elementary coercions into the solving process.
For this purpose, we assume that regular equations and substitutions are annotated with parse trees.
For example, we write $\{ v_1 : R_1 \regEq \alpha_1, \dots, v_n : R_n \regEq \alpha_n \}$ to denote
a set of regular equations $\EE$ where for each $i$ we have that $\EE \turns v_i : R_i$.
Similarly, we write $\{ v_1 : R_1 \mapsto \gamma_1, \dots, v_n : R_n \mapsto \gamma_n \}$
for substitutions.

\begin{defn}[Coercive Solver]
    \bda{c}
  \tlabel{C-Arden} \ \ \ \
  \myirule{R \not\in \alpha} 
     {
   \solveRel{\Config{\psi}{ \Fold\ v : R \regEq s \conc R + \alpha \uplus {\cal E}}}
            {\Config{\psi}{ \Fold\ f_A(v) : R \regEq s^* \conc \alpha \uplus {\cal E}}} }
  \\ \\
  \tlabel{C-Subst} \ \ \ \          
  \myirule{R \not\in \alpha
           \\ \ba{lcl}
           \psi' & = & \{ v : R \mapsto \alpha \}
               \\ & \cup & \{ v : R' \mapsto \alpha' \mid v : R' \mapsto \alpha' \in \psi \wedge R \not\in \alpha' \}
               \\ & \cup & \{ \Fold\ f (f_{\delta\Angle{}} (v)) : R' \mapsto \alpha'' \mid 
                    \Fold\ v : R' \mapsto \delta'\Angle{R} \in \psi \wedge                    
            \\  && \ \ \ \ \ \ \ \ \ \ \ \ \ \ \ \ \ \ \ \ \ \ \ \ \ \ \ \ \ R \not\in \delta'\Angle{\alpha} \wedge
            \\  && \ \ \ \ \ \ \ \ \ \ \ \ \ \ \ \ \ \ \ \ \ \ \ \ \ \ \ \ \ \delta'\Angle{\alpha} \stackrel{f}{\simeq} \alpha'' \}
              \ea
           \\ \ba{lcl}
                \EE' & = & \{ v : R' \regEq \alpha' \mid v : R' \regEq \alpha' \in \EE \wedge R \not\in \alpha' \}
            \\ & \cup & \{ \Fold\ f (f_{\delta\Angle{}} (v)) : R' \regEq \alpha'' \mid 
                    \Fold\ v : R' \regEq \delta'\Angle{R} \in \EE \wedge                    
            \\  && \ \ \ \ \ \ \ \ \ \ \ \ \ \ \ \ \ \ \ \ \ \ \ \ \ \ \ \ \ R \not\in \delta'\Angle{\alpha} \wedge 
            \\  && \ \ \ \ \ \ \ \ \ \ \ \ \ \ \ \ \ \ \ \ \ \ \ \ \ \ \ \ \ \delta'\Angle{\alpha} \stackrel{f}{\simeq} \alpha'' \}
              \ea
           }
          {\solveRel{\Config{\psi}{\Fold\ v : R \regEq \alpha \uplus {\cal E}}}
                    {\Config{\psi'}{\EE'}}
          }
    \eda
\end{defn}  

In the coercive Arden rule, we apply the Arden coercion introduced in Lemma~\ref{le:arden-coercion}.
During substitution we uniformly normalize right-hand sides of equations \emph{and} the codomains of substitutions.
Side condition $R \not\in \delta'\Angle{\alpha}$ guarantees that all occurrences of $R$ are replaced.
Parse trees are transformed by first
applying the substitution coercion followed by the normalization coercion.
Thus, we can transform parse trees of regular equations into parse trees of solutions.

\begin{prop}[Coercive Solving]
\label{prop:coercive-solving}  
  Let $\EE = \{ v_1 : R_1 \regEq \alpha_1, \dots, v_n : R_n \regEq \alpha_n \}$ be a parse tree annotated
  set of regular equations in normal form where $\EE \turns v_i : R_i$ for $i=1,\dots,n$.
  Then, $\solveRels{\Config{\{ \}}{{\cal E}}}{\Config{\psi}{\{ \}}}$ for some substitution $\psi$
  where $\psi = \{ u_1 : R_1 \mapsto s_1, \dots, u_n : R_n \mapsto s_n \}$
  such that $\turns u_i : s_i$ and $\flatten{u_i} = \flatten{v_i}$ for $i=1,\dots,n$.
\end{prop}
\begin{proof}
  Follows immediately from Lemmas~\ref{le:arden-coercion}, \ref{le:normalize-coercion}
  and~\ref{le:subst-coercion}.
  \qed
\end{proof}

\begin{thm}[Unambiguous Solutions]
  \label{theo:reg-solve-unambig}
  Let $\EE$ be a set of non-overlapping equations
 where $\solveRels{\Config{\{ \}}{{\cal E}}}{\Config{\psi}{\{ \}}}$ for some substitution $\psi$.
  Then, for each $R \in {\it dom}(\EE)$ we find that $\psi(R)$ is unambiguous.
\end{thm}
\begin{proof}
  Follows from Propositions~\ref{prop:regeq-unambig} and~\ref{prop:coercive-solving}
  and the fact that coercions are bijective.
  \qed
\end{proof}

\section{Brzozowski's algebraic method}
\label{sec:algebraic-method}

We revisit Brzozowski's algebraic method~\cite{321249} to transform
an automaton into a regular expression.
Based on our results we can show that resulting regular expressions are always unambiguous.

\begin{defn}[Deterministic Finite Automata (DFA)]
  A deterministic finite automaton (DFA) is a 5-tuple $M=(Q, \Sigma, \delta, q_0, F)$
  consisting of a a finite set $Q$ of states,
  a finite set $\Sigma$ of symbols, a transition function $\delta : Q \times \Sigma \rightarrow Q$,
  an initial state $q_0 \in Q$, and a set $F$ of accepting states.
  We say $M$ accepts word $w=x_1 \dots x_n$ if there exists a sequence of states $p_1, \dots, p_{n+1}$
  such that $p_{i+1} = \delta(p_i, x_n)$ for $i=1,\dots,n$,
  $p_1 = q_0$ and $p_{n+1} \in F$.
\end{defn}

Brzozowski turns a DFA  into an equivalent set of (characteristic) regular equations.

\begin{defn}[Characteristic Equations]
  Let $M=(Q, \Sigma, \delta, q_0, F)$ be a DFA.
  We define ${\cal E}_M = \{ R_q \regEq \sum_{x \in \Sigma} x \conc R_{\delta(q,x)} + f(q) \mid q \in Q\}$ where $f(q) = \varepsilon$ if $q \in F$. Otherwise, $f(q) = \phi$.
  We refer to ${\cal E}_M$ as the \emph{characteristic equations} obtained from $M$.
\end{defn}

He suggests solving these equations via Arden's Lemma but the exact details (e.g.~normalization)
are not specified. Assuming we use the solving method specified in Definition~\ref{def:solving}
we can conclude the following.
By construction, characteristic equations are non-overlapping.
From Theorem \ref{theo:reg-solve-unambig} we can derive the following result.

\begin{corollary}
  Solutions obtained from characteristic equations are unambiguous.
\end{corollary}

Instead of a DFA we can also turn a non-deterministic automaton (NFA)
into an equivalent regular expression.
Each $\varepsilon$ transitions is represented by the component $\varepsilon \conc R$.
For two non-deterministic transitions via symbol $x$ to follow states $R_1$
and $R_2$, we generate the component $x \conc R_1 + x \conc R_2$.
Resulting characteristic equations will be overlapping in general.
Hence, we can no longer guarantee unambiguity.

\section{Subtraction and Shuffle}
\label{sec:subtract-intersect}

We introduce direct methods to subtract and shuffle regular expressions.
Instead of turning the regular expressions into a DFA and carrying out the operation
at the level of DFAs, we generate an appropriate set of equations
by employing Brzozowski derivatives.
Solving the equations yields then the desired result.
In essence, our method based on solutions resulting from derivative-based equations
is isomorphic to building a derivative-based DFA from expressions, applying
the product automaton construction among DFAs and then turn the resulting
DFA into an expression via Brzozowski's algebraic method.

For subtraction, equations generated are non-overlapping.
Hence, resulting expressions are also unambiguous.
First, we recall the essential of derivatives before discussing
each operation including some optimizations.

\subsection{Brzozowski's Derivatives}

The \emph{derivative} of a regular expression $r$ with respect to some symbol $x$,
written $\deriv{r}{x}$, is a regular expression for the left quotient of $\lang(r)$ with respect to $x$.
That is, $\lang(\deriv{r}{x}) = \{ w \in\Sigma^* \mid x \conc w \in \lang(r) \}$.
A derivative $\deriv{r}{x}$ can be computed by recursion
over the structure of the regular expression $r$.

\begin{defn}[Brzozowski Derivatives~\cite{321249}]
\bda{ll}
\deriv{\phi}{x} = \phi
 & 
 \deriv{\varepsilon}{x} = \phi
 \\
 \\
  \deriv{y}{x} = \left \{ \ba{ll} \varepsilon & \mbox{if $x=y$}
                          \\      \phi        & \mbox{otherwise}
                          \ea
                          \right.
 & \deriv{r + s}{x} = \deriv{r}{x} + \deriv{s}{x}
\\
\\                          
\deriv{r \conc s}{x} = \left \{ \ba{ll}  \deriv{r}{x} \conc s & \mbox{if $\varepsilon \not\in \lang(r)$ \ \ \mbox{}}
\\  \deriv{r}{x} \conc s + \deriv{s}{x} & \mbox{otherwise}
\ea
\right.
 & 
\deriv{r^*}{x} = \deriv{r}{x} \conc r^*
\eda
\end{defn}

\begin{example}
  The derivative of $(x+y)^*$ with respect to symbol $x$ is $(\varepsilon + \phi) \conc (x+y)^*$.
  The calculation steps are as follows:
  $$
  \deriv{(x+y)^*}{x}
  = \deriv{x+y}{x} \conc (x+y)^*
  = (\deriv{x}{x} + \deriv{y}{x}) \conc (x+y)^*
  = (\varepsilon + \phi) \conc (x+y)^*
  $$

\end{example}

\begin{thm}[Expansion \cite{321249}]
  \label{th:representation}
  Every regular expression $r$ can be represented
  as the sum of its derivatives with respect to all symbols. If $\Sigma =
  \{x_1, \dots, x_n\}$, then
  \bda{c}
  r \semeq x_1 \conc \deriv{r}{x_1} + \cdots + x_n \conc \deriv{r}{x_n}
  \ \mbox{($+ \varepsilon$ if $r$ nullable)}
  \eda  
\end{thm}

\begin{defn}[Descendants and Similarity]
\label{def:desc-sim}
A \emph{descendant} of $r$ is either $r$ itself
or the derivative of a descendant.
We say $r$ and $s$ are \emph{similar}, written $r \sim s$,
if one can be transformed into the other by
finitely many applications of the rewrite rules (Idempotency) $r+r \sim r$,
(Commutativity) $r + s \sim s + r$,
(Associativity) $r + (s + t) \sim (r + s) + t$,
(Elim1) $\varepsilon \conc r \sim r$,
(Elim2) $\phi \conc r \sim \phi$,
(Elim3) $\phi + r \sim r$, and
(Elim4) $r + \phi \sim r$.
\end{defn}
\begin{lem}
 \label{le:sim-re}  
  Similarity is an equivalence relation that respects regular
  expression equivalence: $r \sim s$ implies $r \semeq s$.
\end{lem}
\begin{thm}[Finiteness \cite{321249}]
The elements of the set of descendants of a regular expression belong
to finitely many similarity equivalence classes.
\end{thm}
Similarity rules (Idempotency), (Commutativity), and (Associativity)
suffice to achieve finiteness. Elimination rules are added to
obtain a compact \emph{canonical representative} for
equivalence class of similar regular expressions.
The canonical form is obtained by systematic application
of the similarity rules in Definition~\ref{def:desc-sim}.
We enforce right-associativity of concatenated expressions,
sort alternative expressions according to their size and their first
symbol, and concatenations lexicographically,
assuming an arbitrary total order on $\Sigma$.
We further remove duplicates
and apply elimination rules exhaustively (the details are standard \cite{Grabmayer:2005:UPC:2156157.2156171}). 

\begin{defn}[Canonical Representatives]
\label{def:cnf}  
  For a regular expression $r$, we write $\simp{r}$ to denote
  the canonical representative among all expressions
  similar to $r$.
  We write $\Desc{r}$ for the set of canonical representatives
  of the finitely many dissimilar descendants of $r$.
\end{defn}  

\begin{example}
  We find that $\simp{(\varepsilon + \phi) \conc (x+y)^*} = (x+y)^*$ where $x<y$.
\end{example}

\subsection{Subtraction}

\begin{defn}[Equations for Subtraction]
\label{def:eq-subtract}  
  Let $r, s$ be two regular expressions.
  For each pair $(r',s') \in \Desc{r} \times \Desc{s}$
  we introduce a variable $R_{r',s'}$.
  For each such $R_{r',s'}$ we define an equation of the following form.
  If $\lang(r') = \emptyset$, we set $R_{r',s'} \regEq \phi$.
  Otherwise, $R_{r',s'} \regEq \sum_{x \in \Sigma} x \conc R_{\simp{\deriv{r'}{x}}, \simp{\deriv{s'}{x}}} + t$
  where $t = \varepsilon$ if $\varepsilon \in \lang(r'), \varepsilon \not\in \lang(s')$,
  otherwise $t = \phi$.
  All equations are collected in a set ${\cal S}_{r,s}$.

  Let $\psi = \solve{{\cal S}_{r,s}}$.
  Then, we define $r-s = \psi(R_{r,s})$.
\end{defn}


As the set of canonical derivatives is finite, the set $\solve{{\cal S}_{r,s}}$ is finite as well.
Hence, a solution must exist. Hence, $r-s$ is well-defined.

\begin{lem}
 \label{le:subtract-nf}  
  Let $r, s$ be two regular expressions.
  Then, we find that
  \bda{c}
\lang(r) - \lang(s) \semeq \sum_{x \in \Sigma} x \conc (\lang(\simp{\deriv{r}{x}}) - \lang(\simp{\deriv{s}{x}})) + T
\eda
where $T = \{ \varepsilon \}$ if $\varepsilon \in \lang(r), \varepsilon \not\in \lang(s)$,
otherwise $T = \emptyset$.
\end{lem}
\begin{proof}
  By the Expansion Theorem~\ref{th:representation} and Lemma~\ref{le:sim-re},
  we find that
  $r \semeq \sum_{x \in \Sigma} x \conc \simp{\deriv{r}{x}} + t$
  and $s \semeq \sum_{x \in \Sigma} x \conc \simp{\deriv{s}{x}} + t'$
  where $t = \varepsilon$ if $r$ is nullable. Otherwise, $t=\phi$.
  For $t'$ we find $t' = \varepsilon$ if $s$ is nullable. Otherwise, $t'=\phi$.

  By associativity, commutativity of $+$ and some standard algebraic laws
\bda{c}
(x \conc R) - (x \conc S) \semeq x \conc (R - S)
\\
(x \conc R) - (y \conc S) \semeq x \conc R   \ \ \ \ \ \  \mbox{where $x \not= y$}
\\
R - \phi \semeq R
\\
(R + S) - T \semeq (R - T) + (S - T)
\\
R - (S + T) \semeq (R - S) - T
\eda
the result follows immediately. \qed
\end{proof}

\begin{thm}[Subtraction]
\label{theo:subtract}  
  Let $r, s$ be two regular expressions.
  Then, we find that $r-s$ is unambiguous and $\lang(r-s) \semeq \lang(r) - \lang(s)$.
\end{thm}
\begin{proof}
  By construction, equations are non-overlapping. Unambiguity follows
  from Theorem~\ref{theo:reg-solve-unambig}.
  
  We prove the equivalence claim via a coalgebraic proof method~\cite{DBLP:conf/lata/RotBR13}. We show that the relation
  $\{ (\lang(\psi(R_{r',s'})), \lang(r') - \lang(s')) \mid (r',s') \in \Desc{r} \times \Desc{s} \}$
  is a bisimulation where $\psi = \solve{{\cal S}_{r,s}}$.
  For that to hold two elements are in relation if either (1) they are both nullable,
  or (2) their derivatives, i.e.~taking away the same leading literal, are again in relation.

  Consider a pair $(\lang(\psi(R_{r',s'})), \lang(r') - \lang(s'))$.
  For $\lang(r') = \emptyset$ we have that $R_{r',s'} \approx \phi$.
  The conditions imposed on a bisimulation follow immediately.

  Otherwise, $R_{r',s'}$ is defined by the equation 
  \bda{cr}
      R_{r',s'} \regEq \sum_{x \in \Sigma} x \conc R_{\simp{\deriv{r'}{x}}, \simp{\deriv{s'}{x}}} + t  & \ \ \ \ (E1)
  \eda
  where $t = \varepsilon$ if $\varepsilon \in \lang(r'), \varepsilon \not\in \lang(s')$,
  otherwise $t = \phi$.
  From Lemma~\ref{le:subtract-nf} we can conclude that
  \bda{cr}
    \lang(r) - \lang(s) \semeq \sum_{x \in \Sigma} x \conc (\lang(\simp{\deriv{r}{x}}) - \lang(\simp{\deriv{s}{x}})) + T & \ \ \ \ (E2)
  \eda
  where $T = \{ \varepsilon \}$ if $\varepsilon \in \lang(r), \varepsilon \not\in \lang(s)$,
  otherwise $T = \emptyset$.
  Immediately, we find that if one component of the pair is nullable, the other one must be nullable as well.

  We build the derivative for each component w.r.t.~some literal $x$.
  Given that $\psi$ is a solution and via (E1) and (E2) the resulting derivatives
  are equal to $\lang(\psi(R_{\simp{\deriv{r'}{x}}, \simp{\deriv{s'}{x}}}))$
  and $\lang(\simp{\deriv{r}{x}}) - \lang(\simp{\deriv{s}{x}})$.
  Hence, derivatives are again in relation.
  This concludes the proof. \qed
\end{proof}

\begin{example}
   We consider $r_1 =(x+y)^*$ and $r_2  = (x \conc x)^*$.
   Let us consider first the canonical descendants of both expressions.
   \bda{lclcl}
   \derivC{(x+y)^*}{x} & = & (x+y)^* 
   \\
   \derivC{(x+y)^*}{y} & = & (x+y)^*
   \\
   \derivC{(x \conc x )^*}{x} & = & x \conc (x \conc x)^* & = r_3
   \\
   \derivC{(x \conc x )^*}{y} & = & \phi & = r_4
   \\
   \deriv{x \conc (x \conc x)^*}{x} & = & (x \conc x)^*
   \\
   \deriv{x \conc (x \conc x)^*}{y} & = & \phi
   \eda
 
  The resulting equations are as follows. 
 \bda{lcl}
 R_{1,2}  & = & x \conc R_{1,3} + y \conc R_{1,4} + \phi
 \\
 R_{1,3} & = & x \conc R_{1,2} + y \conc R_{1,4} + \varepsilon
 \\
 R_{1,4} & = & r_1
 \eda
 
 Solving of the above proceeds as follows. We first apply $R_{1,4} = r_1$.
 \bda{lcl}
 R_{1,2}  & = & x \conc R_{1,3} + y \conc r_1 + \phi
 \\
 R_{1,3} & = & x \conc R_{1,2} + y \conc r_1 + \varepsilon
 \eda
 
 Next, we remove the equation for $R_{1,3}$ and apply some simplifications.
 \bda{lcl}
 R_{1,2}  & = & x \conc x \conc R_{1,2} + x \conc y \conc r_1 + x + y \conc r_1 
 \eda
 Via Arden's Lemma we find that $R_{1,2} = (x \conc x)^* \conc (x \conc y \conc r_1 + x + y \conc r_1)$
 and we are done.
 \end{example}

\subsection{Shuffle}

\begin{defn}[Shuffle] 
\label{def:shuffling}
The \emph{shuffle operator} $\shuffle :: \Sigma^* \times \Sigma^* \rightarrow \power(\Sigma^*)$
is defined inductively as follows:
   \bda{lcl}
   \epsilon \shuffle w & = & \{ w\}
 \\   w \shuffle \epsilon & = & \{ w\}
  \\ x \conc v \shuffle y \conc w & = & \{ x \conc u \mid u \in v \shuffle y \conc w \} 
    \cup \{ y \conc u \mid u \in x \conc v \shuffle w \}
  \\ 
 \eda
%
 We lift shuffling to languages by
$$
L_1 \shuffle L_2 = \{ u \mid u \in v \shuffle w \wedge
                                 v \in L_1 \wedge w \in L_2 \}
$$
\end{defn}
For example, we find that $x \conc y \shuffle z = \{ x\conc y \conc z, x \conc z \conc y, z \conc x \conc y \}$.

\begin{defn}[Equations for Shuffling]
\label{def:eq-shuffle}  
  Let $r, s$ be two regular expressions.
  For each pair $(r',s') \in \Desc{r} \times \Desc{s}$
  we introduce a variable $R_{r',s'}$.
  For each such $R_{r',s'}$ we define an equation of the following form.
  If $\lang(r') = \emptyset$, we set $R_{r',s'} \regEq \phi$.
  Otherwise, $R_{r',s'} \regEq \sum_{x \in \Sigma} (x \conc R_{\simp{\deriv{r'}{x}}, s'} + x \conc R_{r', \simp{\deriv{s'}{x}}} ) + t$
  where $t = t_1 + t_2$.
  Expression $t_1 = s'$ if $\varepsilon \in \lang(r')$,
  otherwise $t_1 = \phi$.
  Expression $t_2 = r'$  if $\varepsilon \in \lang(s')$,
  otherwise $t_2 = \phi$.
  All equations are collected in a set ${\cal H}_{r,s}$.

  Let $\psi = \solve{{\cal H}_{r,s}}$.
  Then, we define $r \shuffle s = \psi(R_{r,s})$.
\end{defn}

\begin{lem}
\label{le:shuffle-nf}  
  Let $r, s$ be two regular expressions.
  Then, we find that
  \bda{c}
\lang(r) \shuffle \lang(s) \semeq \sum_{x \in \Sigma} (x \conc (\lang(\simp{\deriv{r}{x}}) \shuffle \lang(s)) + x \conc (\lang(r) \shuffle \lang(\simp{\deriv{s}{x}}))) + T
\eda
where $T = T_1 + T_2$.
$T_1 = s$ if $\varepsilon \in \lang(r)$,
otherwise $T_1 = \phi$.
$T_2 = r$ if $\varepsilon \in \lang(s)$,
otherwise $T_2 = \phi$.
\end{lem}

\begin{thm}[Shuffling]
\label{theo:shuffle}  
  Let $r, s$ be two regular expressions.
  Then, we find that  $\lang(r \shuffle s) \semeq \lang(r) \shuffle \lang(s)$.
\end{thm}
%

\section{Related Works and Conclusion}
\label{sec:related-works}
\label{sec:conclusion}

Our work gives a precise description of solving of regular equations
including a computational interpretation by means of parse tree transformations.
Thus, we can characterize conditions under which regular equations and resulting
regular expressions are unambiguous.

Earlier work by Gruber and Holzer~\cite{DBLP:journals/ijfcs/GruberH15}
gives a comprehensive overview
on the conversion of finite automaton to regular expressions and vice versa.
Like many other works~\cite{c-neumann18,321249}, the algorithmic details
of solving regular equations based on Arden's Lemma are not specified in detail.

Brzozowski's and McCluskey's~\cite{DBLP:journals/tc/BrzozowskiM63}
state elimination method appears to be the more popular and more widespread method.
For example, consider work by
Han~\cite{Han:2013:SEH:2595015.2595018} and in collaboration with Ahn~\cite{Ahn:2009:ISE:1577697.1577720}, as well as work by
Moreira, Nabais and Reis~\cite{DBLP:journals/corr/abs-1008-1656}
that discuss state elimination heuristics to achieve short regular expressions.

Sakarovitch~\cite{DBLP:books/daglib/0023547,automata-rational-expressions} shows
that the state elimination and solving via regular equation methods are isomorphic
and produce effectively the same result. Hence, our (unambiguity) results are transferable
to the state elimination setting. The other way around,
state elimination heuristics are applicable as demonstrated by our implementation.

It is well understood how to build the subtraction and intersection
among DFAs via the product automaton construction~\cite{Hopcroft:2006:IAT:1196416}.
If we wish to apply these operations among regular expressions
we need to convert expressions back and forth to DFAs.
For example, we can convert a regular expression into a DFA
using Brzozowski's derivatives~\cite{321249} 
and then use Brzozowski's algebraic method to convert back
the product automaton to a regular expression.

To build the shuffle among two regular expressions,
the standard method is to (1) build the shuffle derivative-based DFA,
(2) turn this DFA into some regular equations and then
(3) solve these regular equations.
Step (1) relies on the property that the canonical derivatives
for shuffle expressions are finite.

In our own work~\cite{DBLP:journals/jcss/SulzmannT19},
we establish finiteness for several variations of the shuffle operator.
Caron, Champarnaud and Mignot~\cite{DBLP:journals/ita/CaronCM14} and
Thiemann~\cite{DBLP:conf/wia/Thiemann16}
establish finiteness of derivatives for an even larger class
of regular expression operators.

We propose direct methods to build the intersection and
the shuffle among two regular expressions.
For each operation we generate an appropriate set of equations
by employing Brzozowski derivatives. We only rely
on finiteness of canonical derivatives for standard regular expressions.
Solving of these equations then yields the desired expression.
Correctness follows via some simple (co)algebraic reasoning and
we can guarantee that resulting expressions are unambiguous.





\section*{Acknowledgments}

We thank referees for CIAA'18, ICTAC'18 and ICTAC'19
for their helpful comments on previous versions of this paper.

\bibliography{main}

\begin{thebibliography}{10}

\bibitem{Ahn:2009:ISE:1577697.1577720}
Jae-Hee Ahn and Yo-Sub Han.
\newblock Implementation of state elimination using heuristics.
\newblock In {\em Proc.\ of CIAA'09}, pages 178--187. Springer, 2009.

\bibitem{ALMEIDA200793}
Marco Almeida, Nelma Moreira, and Rog\'{e}rio Reis.
\newblock Enumeration and generation with a string automata representation.
\newblock {\em Theoretical Computer Science}, 387(2):93 -- 102, 2007.
\newblock Descriptional Complexity of Formal Systems.

\bibitem{DBLP:conf/focs/Arden61}
Dean~N. Arden.
\newblock Delayed-logic and finite-state machines.
\newblock In {\em 2nd Annual Symposium on Switching Circuit Theory and Logical
  Design, Detroit, Michigan, USA, October 17-20, 1961}, pages 133--151, 1961.

\bibitem{DBLP:conf/ppdp/BrabrandT10}
Claus Brabrand and Jakob~G. Thomsen.
\newblock Typed and unambiguous pattern matching on strings using regular
  expressions.
\newblock In {\em Proc.\ of PPDP'10}, pages 243--254. {ACM}, 2010.

\bibitem{321249}
Janusz~A. Brzozowski.
\newblock Derivatives of regular expressions.
\newblock {\em J. ACM}, 11(4):481--494, 1964.

\bibitem{DBLP:journals/tc/BrzozowskiM63}
Janusz~A. Brzozowski and Edward~J. McCluskey.
\newblock Signal flow graph techniques for sequential circuit state diagrams.
\newblock {\em {IEEE} Trans. Electronic Computers}, 12(2):67--76, 1963.

\bibitem{DBLP:journals/ita/CaronCM14}
Pascal Caron, Jean{-}Marc Champarnaud, and Ludovic Mignot.
\newblock A general framework for the derivation of regular expressions.
\newblock {\em {RAIRO} - Theor. Inf. and Applic.}, 48(3):281--305, 2014.

\bibitem{DBLP:conf/wia/DelgadoM04}
Manuel Delgado and Jos{\'{e}} Morais.
\newblock Approximation to the smallest regular expression for a given regular
  language.
\newblock In {\em Proc.\ of CIAA'04}, pages 312--314. Springer, 2004.

\bibitem{cduce-icalp04}
Alain Frisch and Luca Cardelli.
\newblock Greedy regular expression matching.
\newblock In {\em Proc. of ICALP'04}, pages 618-- 629. Springer, 2004.

\bibitem{Grabmayer:2005:UPC:2156157.2156171}
Clemens Grabmayer.
\newblock Using proofs by coinduction to find "traditional" proofs.
\newblock In {\em Proc.\ of CALCO'05}, pages 175--193. Springer, 2005.

\bibitem{DBLP:journals/ijfcs/GruberH15}
Hermann Gruber and Markus Holzer.
\newblock From finite automata to regular expressions and back - {A} summary on
  descriptional complexity.
\newblock {\em Int. J. Found. Comput. Sci.}, 26(8):1009--1040, 2015.

\bibitem{Han:2013:SEH:2595015.2595018}
Yo-Sub Han.
\newblock State elimination heuristics for short regular expressions.
\newblock {\em Fundam. Inf.}, 128(4):445--462, October 2013.

\bibitem{Hopcroft:2006:IAT:1196416}
John~E. Hopcroft, Rajeev Motwani, and Jeffrey~D. Ullman.
\newblock {\em Introduction to Automata Theory, Languages, and Computation (3rd
  Edition)}.
\newblock Addison-Wesley Longman Publishing Co., Inc., Boston, MA, USA, 2006.

\bibitem{regex-symbolic}
Kenny Z.~M. Lu and Martin Sulzmann.
\newblock {Solving Regular Expression Equations}.
\newblock \\ \verb+http://github.com/luzhuomi/regex-symb+.

\bibitem{DBLP:journals/corr/abs-1008-1656}
Nelma Moreira, Davide Nabais, and Rog{\'{e}}rio Reis.
\newblock State elimination ordering strategies: Some experimental results.
\newblock In {\em Proc.\ of DCFS'10}, volume~31 of {\em {EPTCS}}, pages
  139--148, 2010.

\bibitem{c-neumann18}
Christoph Neumann.
\newblock Converting deterministic finite automata to regular expressions.
\newblock
  \verb+http://citeseerx.ist.psu.edu/viewdoc/summary?doi=10.1.1.85.2597+, March
  2005.

\bibitem{DBLP:conf/lata/RotBR13}
Jurriaan Rot, Marcello~M. Bonsangue, and Jan J. M.~M. Rutten.
\newblock Coinductive proof techniques for language equivalence.
\newblock In {\em Proc.\ of LATA'13}, pages 480--492. Springer, 2013.

\bibitem{DBLP:books/daglib/0023547}
Jacques Sakarovitch.
\newblock {\em Elements of Automata Theory}.
\newblock Cambridge University Press, 2009.

\bibitem{automata-rational-expressions}
Jacques Sakarovitch.
\newblock Automata and rational expressions.
\newblock \\ \verb+https://arxiv.org/abs/1502.03573+, 2015.

\bibitem{DBLP:conf/flops/SulzmannL14}
Martin Sulzmann and Kenny Zhuo~Ming Lu.
\newblock {POSIX} regular expression parsing with derivatives.
\newblock In {\em Proc.\ of FLOPS'14}, pages 203--220. Springer, 2014.

\bibitem{DBLP:journals/ijfcs/SulzmannL17}
Martin Sulzmann and Kenny Zhuo~Ming Lu.
\newblock Derivative-based diagnosis of regular expression ambiguity.
\newblock {\em Int. J. Found. Comput. Sci.}, 28(5):543--562, 2017.

\bibitem{DBLP:journals/jcss/SulzmannT19}
Martin Sulzmann and Peter Thiemann.
\newblock Derivatives and partial derivatives for regular shuffle expressions.
\newblock {\em J. Comput. Syst. Sci.}, 104:323--341, 2019.

\bibitem{DBLP:conf/wia/Thiemann16}
Peter Thiemann.
\newblock Derivatives for enhanced regular expressions.
\newblock In {\em Proc.\ of CIAA'16}, pages 285--297. Springer, 2016.

\end{thebibliography}

\pagebreak

\appendix

\section{Proofs}

\subsection{Proof of Proposition~\ref{prop:regeq-unambig}}

\begin{proof}
For non-overlapping equations there can be at most one $v$
such that ${\cal E} \turns v : R $.
Suppose ${\cal E} \turns v: R$ where $R \regEq x_1 \conc R_1 + \dots + x_n \conc R_n + t$.
Recall that $+$ is right-associative.
From ${\cal E} \turns v : R$ we conclude that ${\cal E} \turns v' : x_1 \conc R_1 + \dots + x_n \conc R_n + t$
where $v = \Fold \ v'$ for some $v'$.
For $\flatten{v'} = \varepsilon$ (empty word), $t$ must be nullable. Based
on the choice for $t$, we must have that $t = \varepsilon$. Hence, the choice for $v'$ is fixed.
Consider $\flatten{v'} = x \conc w$ for some literal $x$ and word $w$.
Again the choice for $v'$ is fixed because due to non-overlapping there is at most one $i$
such that $x = x_i$.
\qed
\end{proof}

\subsection{Proof of Lemma~\ref{le:shuffle-nf}}

\begin{proof}

  We employ the following algebraic laws.
  \bda{c}
  R \shuffle S \semeq S \shuffle R
  \\
  R \shuffle \emptyset \semeq \emptyset
  \\
  R \shuffle \varepsilon \semeq R
  \\
  (x \conc R) \shuffle (y \conc S) \semeq (x \conc (R \shuffle (y \conc S))) +
  (y \conc ((x \conc R) \shuffle S))
  \\
  (R + S) \shuffle T \semeq (R \shuffle T) + (S \shuffle T)
  \eda
 \qed
\end{proof}

\subsection{Proof of Theorem~\ref{theo:shuffle}}

\begin{proof}
  Similar to the proof of Theorem~\ref{theo:subtract}.

  Unambiguity may no longer hold because
equations to compute $r \shuffle s$ are overlapping.
  \qed
\end{proof}

\section{Parsing with Regular Equations}
\label{sec:parse-reg-eq}

We build a parser following
the scheme of a derivative-style regular expression parser.
Soundness results reported in~\cite{DBLP:conf/flops/SulzmannL14}
carry over to the extended setting.
That is, if a parse exists the parser will succeed.
It is possible to compute all parse trees following
the scheme outlined in~\cite{DBLP:journals/ijfcs/SulzmannL17}.
For brevity, we omit the details.

\begin{defn}[Regular Expressions and Equations Derivatives]
 We assume a fixed set $\EE$ of equations.
\bda{ll}
\deriv{\phi}{x} = \phi
 & 
 \deriv{\varepsilon}{x} = \phi
 \\
 \\
  \deriv{y}{x} = \left \{ \ba{ll} \varepsilon & \mbox{if $x=y$}
                          \\      \phi        & \mbox{otherwise}
                          \ea
                          \right.
 & \deriv{\alpha + \beta}{x} = \deriv{\alpha}{x} + \deriv{\beta}{x}
\\                          
\deriv{\gamma_1 \conc \gamma_2}{x} = \left \{ \ba{ll}  \deriv{\gamma_1}{x} \conc s & \mbox{if $\varepsilon \not\in \lang(\gamma_1)$}
\\  \deriv{\gamma_1}{x} \conc \gamma_2 + \deriv{\gamma_2}{x} & \mbox{otherwise}
\ea
\right.
 & 
\deriv{r^*}{x} = \deriv{r}{x} \conc r^*
\\
\deriv{R}{x} = \deriv{\alpha}{x} \ \ \ \mbox{where $R \regEq \alpha \in \EE$}
\eda
\end{defn}

All standard results, e.g.~expansion, carry over to the extended setting.

\begin{defn}[Empty Parse Trees]
  We assume a fixed set $\EE$ of equations.
  The (partial) function $\mkEmp{\cdot}$ computes an empty parse tree for regular expressions
  and equations.
\begin{haskell}
\mkEmp{\epsilon} = \Eps
\\
\mkEmp{r^*} = [] 
\\
\mkEmp{ \gamma_1 \conc \gamma_2 } = \Seq{(\mkEmp{\gamma_1})}{(\mkEmp{\gamma_2})}
\\
\mkEmp{ \alpha_1 + \alpha_2 } =  \IF\ \varepsilon \in \lang(\alpha_1) \ \THEN \Left\ (\mkEmp{\alpha_1}) \ELSE\ \Right\ (\mkEmp{\alpha_2})
\\
\mkEmp{R} = \mkEmp{\alpha} \ \ \ \WHERE\ R \regEq \alpha \in \EE
\end{haskell}
\end{defn}

\begin{defn}[Injection]
  ~~ 
\begin{haskell}
\inj_{\deriv{r^*}{x}} (v,vs) =  (\inj_{\deriv{r}{x}} ~v) : vs \\
\inj_{\deriv{(\gamma_1 \conc \gamma_2)}{x}} = 
 \hsbody{\lambda v.
  \hsalign{
     \CASE\ v ~\OF \\
    \Seq{v_1}{v_2} \arrow \Seq{(\inj_{\deriv{\gamma_1}{x}} ~v_1)}{v_2} \\
    \Left \ (\Seq{v_1}{v_2}) \arrow \Seq{(\inj_{\deriv{\gamma_1}{x}} ~v_1)}{v_2} \\
    \Right \ v_2 \arrow \Seq{\mkEmp{\gamma_1}}{(\inj_{\deriv{\gamma_2}{x}} ~v_2)} 
  }} \\
\inj_{\deriv{(\alpha_1+\alpha_2)}{x}} = 
  \hsbody{\lambda v.
  \hsalign{
     \CASE\ v ~\OF \\
     \Left \ v_1 \arrow \Left ~(\inj_{\deriv{\alpha_1}{x}} ~v_1) \\
     \Right \ v_2 \arrow \Right ~(\inj_{\deriv{\alpha_2}{x}} ~v_2)
  }} \\
\inj_{\deriv{x}{x}} \Eps = \Sym x
\\
\inj_{\deriv{R}{x}} v = \Fold\ (\inj_{\deriv{\alpha}{x}} v) \ \ \ \ \WHERE R \regEq \alpha \in \EE  
\end{haskell}
\end{defn}

\begin{defn}[Parsing]
  ~~
\begin{haskell}  
  parse = \lambda \gamma. \lambda w.
  \hsalign{
    \CASE\ w \ \OF \\
    \varepsilon \rightarrow \mkEmp{\gamma}
    \\
    x \conc w \rightarrow \inj_{\deriv{\gamma}{x}} (parse \ \deriv{\gamma}{x} \ w)
  }
\end{haskell}
\end{defn}

\section{Coercion Semantics}
\label{sec:coercion-semantics}

We assume function and constructor application to be left associative.
We write $k \ v_1 \dots \ v_m$ as a shorthand
for $( \dots (k \ v_1) \dots ) \ v_m$.

The coercions defined in this paper satisfy the following conditions.
Coercions are first-order where recursion, if any, always takes place
at the outermost level. There are no mutually recursive coercions.
All patterns in a case expression are disjoint.
Then, the following rules are sufficient for evaluation.

\begin{defn}[Big-Step Operational Semantics]
  \bda{cc}
  \tlabel{Lam} &
  \myirule{[x \mapsto v] c \Downarrow v'}
          {(\lambda x.c) \ v \Downarrow v'}
  \\
  \\
  \tlabel{Rec} &
  \myirule{[x \mapsto \REC\ x.c] c \Downarrow v'}
          {\REC\ x.c \Downarrow v'}
  \\
  \\
  \tlabel{Case} &
  \myirule{\exists i \in \{1,\dots,n\}
           \\ pat_i = k \ y_1 \dots y_m
           \\ \mbox{$[$} y_1 \mapsto v_1, \dots, y_m \mapsto v_m \mbox{$]$} c_i \Downarrow v'
        }
        {\CASE\ k \ v_1 \dots v_m \ \OF\ \mbox{$[$} pat_1 \Rightarrow c_1, \ldots, pat_n \Rightarrow c_n \mbox{$]$} \Downarrow v'}
  \eda
\end{defn}

In rule~\tlabel{Lam} we write $[x \mapsto v] c$ to denote replacing
the lambda-bound variable $x$ with argument $v$ in the body $c$.
Arguments $v$ are always parse trees.
Rule~\tlabel{Rec} unfolds the recursive function to carry out the evaluation.
Rule~\tlabel{Case} matches the incoming value against one of the patterns
and then applies to pattern binding to the selected case.
As we assume that patterns are disjoint, the choice of pattern $pat_i$
is unique.

\section{Implementation}
\label{sec:implementation}

We report on a Haskell implementation~\cite{regex-symbolic} of
the regular equation solving algorithm.
We choose Haskell as we can easily derive an implementation
following the formal description in the earlier section.

Given ${\cal E} = \{R_1
\regEq \alpha_1, \dots, R_n \regEq \alpha_n \}$, the algorithm solves
regular equations according to a predefined order among the regular language symbols
$R_i$s. For instance, assuming $R_1 \prec R_2
... \prec R_n$, our implementation solves $R_1 \regEq \alpha_1$ first, then $R_2
\regEq \alpha_2$, and finally $R_n \regEq \alpha_n$ as the default order.
The size of the resulting expressions is sensitive to the order
equations are solved.

This issue has been addressed in the state elimination setting
where via some heuristics certain states are favored to obtain
short regular expressions.
Moreira, Nabais, and Reis~\cite{DBLP:journals/corr/abs-1008-1656} discuss a number of heuristics.
Two of them were reported to be the most effective for a wide set of test cases.
We show how to adapt these heuristics to the regular equation solving setting.

\subsection{Delgado and Morais Heuristics}

In \cite{DBLP:conf/wia/DelgadoM04}, Delgado and Morais propose 
a strategy to choose a state
with the lowest weight to be eliminated in every iteration of the
state elimination method. Given a state, 
the weight function measures the weighted sum of the
regular expressions associated with incoming transitions, the ones
associated with the out-going transitions and those associated
with the loop transitions. We adapt this heuristic in the regular
expression equation solving by associating each regular language
symbol with a weight function.

\begin{defn} [In-coming, out-going and looping regular expressions]
Let ${\cal E} = \{R_1 \regEq \alpha_1, \dots, R_n \regEq \alpha_n
\}$. We define the in-coming regular expressions, out-going regular
expressions and looping regular expression of the regular language symbol
$R_i $ as
\bda{lll}
inRE({\cal E}, R_i) & = & \{ r | R_j \regEq \alpha_j \in {\cal E}
\wedge
R_i \neq R_j \wedge 
\alpha_j = r \conc R_i + \alpha_j'~ \mbox{for some} ~\alpha_j' \} \\ \\

outRE({\cal E}, R_i) & = & \{ r_1, ... , r_m\}~ \mbox{where} ~ R_i
\regEq  r_1 \conc R_1' + ...  + r_{m-1} \conc R_{m-1}' + r_m + r_{m+1}
\conc R_i  \in {\cal E} \\ \\

loopRE({\cal E}, R_i) & = &  r_{m+1}  ~\mbox{where} ~ R_i
\regEq  r_1 \conc R_1' + ...  + r_{m-1} \conc R_{m-1}' + r_m + r_{m+1}
\conc R_i  \in {\cal E} \\ \\

\eda
\end{defn}

\begin{defn}[Weight function] 
Let ${\cal E} = \{R_1 \regEq \alpha_1, \dots, R_n \regEq \alpha_n
\}$.  We define the weight function of a regular language symbol $R_i$
w.r.t ${\cal E}$ as 

\bda{lll}
 W({\cal E}, R_i)   & = & (in_i -1) * \Sigma_{r \in outRE({\cal E},
   R_i)} | r | \\ 
& + & (out_i -1) * \Sigma_{r \in inRE({\cal E}, R_i)} | r | \\ 
 & +  &  (in_i * out_i - 1) *  | loop_i | 
\eda 
where $ in_i =|inRE({\cal E}, R_i))| $, $out_i = |outRE({\cal E}, R_i)|$
and $loop_i = |loopRE({\cal E}, R_i)|$. Given a regular expression
$r$, $| r |$ denotes the alphabetic width of $r$.
\end{defn}

In each solving step, we first apply the weight function to identify the
$R_i \regEq \alpha_i$ with the lowest $W({\cal E}, R_i)$ value, then we apply rules \tlabel{Arden} and
\tlabel{Subst} to solve $R_i \regEq \alpha_i$.

\subsection{Cycle Counting Heuristics}

Cycle counting is another effective heuristic strategy reported in
\cite{DBLP:journals/corr/abs-1008-1656}. The idea of this heuristics is
to treat the regular equations as a directed graph and use the number
of unique cycles as the weight.  In each iteration, the $R_i \regEq
\alpha_i$ with the least number of cycle counts will be selected.

\subsection{Benchmarks}
\fig{f:benchmarks}{Default solving order vs Delgado-Morais heuristics
  vs Cycle Counting Heuristics}{
\begin{center}
\begin{tabular}{ |c|c|c|c| } 
 \hline
  Num. of equations & Num. of symbols & Delgado-Morais vs default &
                                                               Cycle-Count vs default\\ 
 \hline
  5 & 5 & 0.1524 & 0.3464 \\ 
  10 & 5 & 0.0078 & 0.0335 \\ 
  10 & 10 & 0.0093 & 0.0679 \\
 \hline
\end{tabular}
\end{center}
}

Figure~\ref{f:benchmarks} reports benchmark results
where we compare the above-mentioned heuristics against the default
solving order. For each heuristics we measure the average regular
expression size ratio between those obtained from the
heuristics strategy and those computed using the default solving
order. Each test set consists of 1000 test cases generated uniformly by the DFA enumeration framework
\cite{ALMEIDA200793}. We then convert the DFAs into regular
expression equation sets and feed them into the solvers. As observed
from the benchmark results, both heuristic strategies yield more compact regular
expressions as compared to the default solving orders. Similar results
were reported for the state elimination method in \cite{DBLP:journals/corr/abs-1008-1656}.
The source code of the benchmark can be found in \cite{regex-symbolic}.


\section{Observations}

Example~\ref{ex:solve1} suggests that if there is a mutual dependency among
equations and equations are reordered, the solutions we obtain, albeit semantically the same,
may differ syntactically.
We verify this claim below. In fact, we can also show that in the absence
of mutual dependencies, the solution obtained is independent of the
order in which equations are solved.

\begin{defn}
  We say a set ${\cal E}$ of $n$ equations is in \emph{strict order}
  if equations can be sorted such that for each $R_i \regEq s_i \conc R_i + \alpha_i$
  we have that none of the variables $R_{i+1}, \dots, R_n$ appear in $\alpha_i$.
\end{defn}
Assuming equivalences such as $\emptyset \conc R \semeq \emptyset$, it is easy to make any set ${\cal E}$ non-strict.
Hence, we assume that in the initial set ${\cal E}$ for each subcomponent $s \conc R$
we have that $\lang(s) \not \semeq \emptyset$.

\begin{prop}[Order Independent Syntactic Form of Solutions]
 \label{prop:order-independent}  
  Let ${\cal E}$ be a set of equations in strict order.
  Let ${\cal E'}$ be a permutation of the equations ${\cal E}$.~\footnote{Recall that we treat the set of equations like a list.}
  Then, we find that $\psi = \solve{{\cal E}}$ and $\psi' = \solve{{\cal E'}}$
  for some $\psi, \psi'$ where for each $R \in {\it dom}({\cal E})$ we have that
  $\psi(R)$ and $\psi'(R)$ are syntactically the same.
\end{prop}
\begin{proof}
  We assume equations in ${\cal E}$ are enumerated in strict order.
  Consider the $i$th equation $R_i \regEq s_i \conc R_i + \alpha_i$.
  Recall that normalization always achieves this form.
  Due to the strict order assumption, the Arden substitution $[R_i \mapsto s_i^* \conc \alpha_i]$
  only affects later equations $R_j \regEq s_j \conc R_j + \alpha_j$ where $i < j$.
  The shape of equations guarantees that only the $\alpha_j$ component will be affected.
  
  Another consequence of the strict order assumption is that,
  if we solve the equation connected to $R_j$, the resulting Arden
  substitution $[R_j \mapsto s_j^* \conc \alpha_j]$
  will not affect any of the earlier equations (solving steps). 
  Hence, we conclude that solving of equations in strict order yields the same result
  regardless of the order in which equations are solved. \qed
\end{proof}

\begin{prop}[Order Dependent Syntactic Form of Solutions]
\label{prop:dependent-syntactic-form}  
  Let ${\cal E}$ be a set of equations not in strict order.
  Then, we find a permutation ${\cal E'}$ of ${\cal E}$
  such that $\psi = \solve{{\cal E}}$ and $\psi' = \solve{{\cal E'}}$
  for some $\psi, \psi'$ where for some $R \in {\it dom}({\cal E})$ we have that
  $\psi(R)$ and $\psi'(R)$ are syntactically different.
\end{prop}
\begin{proof}
  As the strict order condition is violated,
  we must encounter (either initially or during solving) two equations of the following form
  \bda{ll}
  R_i \regEq s_i \conc R_i + s_j' \conc R_j + \alpha_i & (E_i)
  \\
  R_j \regEq s_j \conc R_j + s_i' \conc R_i + \alpha_j & (E_j)
  \eda
  where neither $R_i$ nor $R_j$ appear in $\alpha_i$ or $\alpha_j$.
  
  Suppose, we solve $(E_i)$ first this leads to the Arden substitution
  $[R_i \mapsto s_i^* \conc (s_j' \conc R_j + \alpha_i)]$.
  Applied on $(E_j)$ we find
  \bda{c}
    R_j \regEq (s_j + s_i' \conc s_i^* \conc s_j') \conc R_j + (s_i' \conc s_i^* \conc \alpha_i + \alpha_j)
  \eda
  We can conclude that the final solution will be of the following form
  \bda{c}
   \psi = [R_i \mapsto s_i^* \dots, R_j \mapsto (s_j + s_i' \conc s_i^* \conc s_j') \dots, \dots]
  \eda

  Reversing the order of equations, given preference to $(E_j)$, results in a final solution
  of the following form
  \bda{c}
   \psi' = [R_j \mapsto s_j^* \dots, R_i \mapsto (s_i + s_j' \conc s_j^* \conc s_i') \dots, \dots]
  \eda
  We immediately find that $\psi(R_i)$ and $\psi'(R_i)$ are syntactically different.
  This concludes the proof. \qed
\end{proof}

For observations concerning  differences in terms of structural complexity
in case of order dependent syntactic forms of solutions
we refer to~\cite{DBLP:journals/ijfcs/GruberH15} for details.

\section{Intersection}
 
\begin{defn}[Equations for Intersection]
  Let $r, s$ be two regular expressions.
  For each pair $(r',s') \in \Desc{r} \times \Desc{s}$
  we introduce a variable $R_{r',s'}$.
  For each such $R_{r',s'}$ we define an equation of the following form.
  $R_{r',s'} \regEq \phi$ if $\lang(r') = \emptyset$ or $\lang(s') = \emptyset$.
Otherwise, $R_{r',s'} \regEq \sum_{x \in \Sigma} x \conc R_{\simp{\deriv{r'}{x}}, \simp{\deriv{s'}{x}}} + t$
where $t = \varepsilon$ if $\varepsilon \in \lang(r'), \varepsilon \in \lang(s')$, otherwise $t=\phi$.
  All equations are collected in a set ${\cal I}_{r,s}$.

  Let $\psi = \solve{{\cal I}_{r,s}}$.
  Then, we define $r \cap s = \psi(R_{r,s})$.
\end{defn}

The above definition is well-defined. Same arguments as for Definition~\ref{def:eq-subtract}  apply.

\begin{lem}
 \label{le:intersect-nf}  
  Let $r, s$ be two regular expressions.
  Then, we find that
  \bda{c}
\lang(r) \cap \lang(s) \semeq \sum_{x \in \Sigma} x \conc (\lang(\simp{\deriv{r}{x}}) \cap \lang(\simp{\deriv{s}{x}})) + T
\eda
where $T = \{ \varepsilon \}$ if $\varepsilon \in \lang(r), \varepsilon \in \lang(s)$,
otherwise $T = \emptyset$.
\end{lem}
\begin{proof}
  Similar to proof of Lemma~\ref{le:subtract-nf}.
  We make use of the following algebraic laws.
  
 \bda{c}
 R \cap \phi \semeq \phi
 \\
 (R + S) \cap T \semeq (R \cap T) + (S \cap T)
 \\
 R \cap (S + T) \semeq (R \cap S) + (R \cap T)
 \\
 (x \conc R) \cap (x \conc S) \semeq x \conc (R \cap S)
 \\
 (x \conc R) \cap (y \conc S) \semeq \phi \ \ \ \ \ \  \mbox{where $x \not= y$}
 \\
 \varepsilon \cap (x \conc R) \semeq \phi
 \\
 (x \conc R) \cap \varepsilon \semeq \phi
 \eda  
\qed
\end{proof}

Applying similar reasoning as for Theorem~\ref{theo:subtract}
we obtain the following result.

\begin{thm}[Intersection]
  Let $r, s$ be two regular expressions.
  Then, we find that $r \cap s$ is unambiguous and $\lang(r \cap s) \semeq \lang(r) \cap \lang(s)$.
\end{thm}
For intersection, all equations are in non-overlapping form and thus we can establish unambiguity.
For the remaining part, we follow  the structure of the proof of Theorem~\ref{theo:subtract}.

\section{Optimizations}

Employing additional similarity rules significantly can reduce
the size of the equations and yields then often more smaller solutions.
We observe this effect in case of subtraction.

\begin{example}
  We build the subtraction
  among regular expression $x^* \conc y^*$ and $x^*$.
  Without any optimizations (using plain canonical derivatives via similarity rules (Idempotency), (Commutativity), and (Associativity)),
  we generate 24 equation and obtain the (subtracted) result
  \bda{c}
  ((x \conc (x \conc ((x^* \conc y) \conc (y^* \conc \varepsilon)))) + (x \conc (y \conc (y^* . \varepsilon))) + (y \conc (y \conc (y^* \conc \varepsilon))) + (y \conc \varepsilon))
  \eda
  By employing further similarity rules such as (Elim1) etc, we generate 6 equations and obtain the result
  \bda{c}
   ((x^* \conc y) \conc y^*)
  \eda
\end{example}

\end{document}